\documentclass[conference]{IEEEtran}
\usepackage{bm}
\usepackage{latexsym}
\usepackage{graphicx}
\usepackage[mathscr]{euscript}
\usepackage{amsmath}
\usepackage{amssymb}
\usepackage{mathrsfs}
\usepackage{fontenc}
\usepackage[utf8]{inputenc}
\usepackage{epstopdf}
\usepackage{amsthm}
\usepackage{amsopn}
\usepackage{amsfonts}
\IEEEoverridecommandlockouts
\newtheorem{theorem}{Theorem}

\newtheorem{proposition}{Proposition}
\newtheorem{corollary}{Corollary}

\begin{document}

\title{Power Distribution of Device-to-Device Communications in Underlaid Cellular Networks}
\author{
\IEEEauthorblockN{Morteza Banagar, Behrouz Maham, \IEEEmembership{Senior Member, IEEE}, Petar Popovski, \IEEEmembership{Senior Member, IEEE}, and\\Francesco Pantisano, \IEEEmembership{Member, IEEE}}\\
\thanks{Morteza Banagar is with the department of ECE, University of Tehran, Iran. Behrouz Maham is with Nazarbayev University, Astana, Kazakhstan. Petar Popovski is with the Department of Electronic Systems, Aalborg University, Aalborg, Denmark. Francesco Pantisano is with JRC -- Joint Research Centre, European Commission, Ispra, Italy.}\vspace{-1.5cm}}
\maketitle


\begin{abstract}
Device-to-device (D2D) communications have recently emerged as a novel transmission paradigm in wireless cellular networks. D2D transmissions take place concurrently with the usual cellular connections, and thus, controlling the interference brought to the macro-cellular user equipment (UE) is of vital importance. In this paper, we consider the uplink transmission of a tier of D2D users that operates as an underlay for the traditional cellular network. Using network model based on stochastic geometry, we derive the \emph{equilibrium} cumulative distribution function (CDF) of the D2D transmit power. Considering interference-limited and relatively lossy environment cases, closed form equations are derived for the power CDF. Finally, a tight closed-form upper-bound for the derived power distribution is proposed, and the analytical results are validated via simulation.
\end{abstract}
\begin{IEEEkeywords}
D2D communication, power distribution, underlaid cellular network, stochastic geometry.
\end{IEEEkeywords}

\IEEEpeerreviewmaketitle
\section{Introduction} \label{SectIntro}

The basic motivation for using Device-to-device (D2D) communication in cellular systems is to enable communication among proximate devices with low latency and low energy consumption, while improving the spectral efficiency \cite{Mag_Doppler_D2D,Mag_Lin_AnOverview, W_Corson_Toward}. Furthermore, it is an efficient paradigm for offloading traffic from macro base stations (BS). D2D communication among proximate user equipment (UE) can occur with or without coordination from the infrastructure (BSs). In this paper we treat the case in which the D2D transmissions are not coordinated by the BS, and instead each UE selects its transmit power in an autonomous way. 

Statistical analysis of the transmit power of D2D users has a key role in determining the capacity and performance of D2D-underlaid cellular networks. In this paper, we first analyze the distribution of the transmit power of D2D users in an underlaid cellular network using stochastic geometry. The spectrum sharing scenario in our model is underlay in-band, i.e., macro-cellular UEs and D2D users concurrently operate over the same licensed spectrum band. In such a model, the transmission of D2D users may interfere with the macro-cellular UEs, making our analysis differ from \cite{Lett_Erturk_Distributions}, which considers an overlay scenario. We derive the cumulative distribution function (CDF) of the transmit power of D2D users, which represents an ``equilibrium" distribution, i.e., a distribution of the D2D transmit powers when all D2D links have attained the target signal-to-interference-plus-noise ratio (SINR). In addition, we consider two special cases of the transmit power CDF, i.e., interference-limited and relatively lossy environments. A tight closed-form upper-bound for the power distribution is also computed.
Once the equilibrium distribution is found, the question is which power control algorithm should be used by each node  in order to attain it. The recent literature features several power control algorithms in D2D underlay cellular networks \cite{CoRR_Asadi_D2DSurvey, CoRR_Lee_Power}. The location of UEs (and sometimes BSs) are random in actual networks, and thus, a suitable stochastic model should be used in order to analyze such random networks. In this paper, we use the Foschini-Miljanic power control algorithm \cite{Trans_Foschini_ASimple}, in which the power allocated to each node is updated in a distributed fashion so as to reach the SINR threshold at all D2D links. Simulation results demonstrate the validity of our analysis.


\section {System Model} \label{SectSysMod}

Consider a set of macro-cellular UEs and a set of D2D users operating on the same frequency band. The locations of macro-cellular UEs and D2D users are modeled as independent homogeneous Poisson point processes (PPP) $\Phi_{\rm C}$ and $\Phi_{\rm D}$ with constant intensities $\lambda_{\rm C}$ and $\lambda_{\rm D}$, respectively. The transmitting powers of macro-cellular UEs are assumed to be independent and identically distributed (i.i.d.) and independent of the transmitting powers of D2D users, which are also considered to be i.i.d. \cite{Lett_Erturk_Distributions}. All links are assumed to be affected by i.i.d. Rayleigh fading which are denoted by $h$. Hence, the fading power has an inversely exponential distribution with mean $\mu$. The target SINR for establishing a connection in cellular and D2D domains are supposed to be $\beta_{\rm C}$ and $\beta_{\rm D}$, respectively. Path loss has a standard form of $\left\|x\right\|^{-\alpha}$ for a transmitter-receiver distance of $x$ and a path-loss exponent of $2<\alpha<6$.

We suppose that each D2D receiver only connects to the nearest D2D transmitter and each D2D transmitter communicates with only one D2D receiver. Without loss of generality, we focus our analysis on a {\em typical} D2D link located at the origin $\mathcal{O}$ and denote it as subscript $0$. In this setting, the received power in a D2D link can be expressed as
\begin{equation} \label{Pr0}
P_{\rm r}=P_{0,{\rm D}}\left|h_{0,{\rm D}}\right|^2\left\|x_{0,{\rm D}}\right\|^{-\alpha},
\end{equation}
where $P_{0,{\rm D}}$, $\left|h_{0,{\rm D}}\right|^2$, and $\left\|x_{0,{\rm D}}\right\|$ are the power of the typical D2D transmitter, the fading power in the typical D2D transceiver channel, and the distance between the typical D2D transmitter and receiver, respectively. The received power from all other macro-cellular UEs (cellular interference) and D2D transmitters (D2D interference) at the typical D2D receiver can be written as follows:
\begin{align} \label{ICID}
I_{\rm C} &= \sum_{i\in \Phi_{\rm C}}P_{i,{\rm C}}\left|h_{i,{\rm C}}\right|^2\left\|x_{i,{\rm C}}\right\|^{-\alpha}, \\
I_{\rm D} &= \sum_{i\in \left\{\Phi_{\rm D}\backslash\mathcal{O}\right\}}P_{i,{\rm D}}\left|h_{i,{\rm D}}\right|^2\left\|x_{i,{\rm D}}\right\|^{-\alpha},
\end{align}
where $P_{i}$, $\left|h_{i}\right|^2$, and $\left\|x_{i}\right\|$ denote the transmit power, the fading power, and the locations of macro-cellular UEs (subscript C) and D2D users (subscript D), respectively. Since the transmit power, the fading power, and the location of macro-cellular UEs and D2D users are all independent of each other, $I_{\rm C}$ and $I_{\rm D}$ are also independent. The noise power is denoted by $\sigma^2$. Furthermore, we consider no out-of-cell interference in our analysis. The SINR at the typical D2D receiver is given by:
\begin{equation}
\mathrm{SINR}=\frac{P_{\rm r}}{I_{\rm D}+I_{\rm C}+\sigma^2}.
\end{equation}


The Foschini-Miljanic algorithm \cite{Trans_Foschini_ASimple} is used as the power controlling method, which reaches the SINR threshold at all links in a distributed fashion. The power of the $i^{\rm th}$ D2D transmitter at time $k+1$ is given by:
\begin{equation} \label{UpdateFM}
P_{i,{\rm D}}\left(k+1\right)=\left(1-\gamma\right)P_{i,{\rm D}}\left(k\right)\left[ 1+\frac{\gamma}{1-\gamma}\frac{\beta_{\rm D}}{\beta_i} \right],
\end{equation}
where $\gamma$ is the convergence rate constant and $\beta_i$ is the SINR of the $i^{\rm th}$ D2D link. Clearly, the algorithm converges when $\beta_{\rm D}=\beta_i$. Each node should know the value of $\gamma$, $\beta_{\rm D}$, and $\beta_i$ at each time $k$ in order to run the algorithm. Equation \eqref{UpdateFM} is used for power allocation in our simulations.

\section{Power Distribution of D2D-Underlaid Cellular Network} \label{SectPowDist}

In order to establish a connection between a D2D pair, the received SINR at the D2D receiver must be greater than a minimum SINR threshold $\beta_{\rm D}$. In this setting, we aim at minimizing the power allocation at each D2D link, while meeting the minimum SINR requirement $\beta_D$, i.e.,
\begin{equation} \label{EqSINR}
\mathrm{SINR} \geq \beta_{\rm D}  \Longleftrightarrow P_{0,{\rm D}} \geq \frac{\beta_{\rm D} \left\|x_{0,{\rm D}}\right\|^{\alpha} \left(I_{\rm D}+I_{\rm C}+\sigma^2\right)}{\left|h_{0,{\rm D}}\right|^2}.
\end{equation}
The following theorem gives the general distribution for the unconstrained D2D transmit power.

\begin{theorem} \label{Theorem1}
In a D2D-underlaid cellular network with multiple macro-cellular UEs and D2D users, the converged CDF of the transmit power of D2D users is
\begin{equation} \label{EqDistD2D}
\mathbb{P}\left\{P_{\rm D}\leq p\right\} = \mathbb{P}\left\{P_{0,{\rm D}}\leq p\right\} =\int_0^{\infty} {\rm e}^{-k_1 r- k_2 r^{\alpha/2}} \, \mathrm{d}r,
\end{equation}
where
\small
\begin{equation} \label{Constants1}
k_1 \!=\! 1+ \frac{1}{\mathrm{sinc}\!\left(2/{\alpha}\right)}\! \left(\!\frac{\beta_{\rm D}}{\mu p}\!\right)^{\!\!2/{\alpha}}\!\!\! \left(\!\!E_{\rm D}\!+\!\!\left(\!\frac{\lambda_{\rm C}}{\lambda_{\rm D}}\!\right)\!\! E_{\rm C}\!\right)\!\!,~~
k_2 \!=\! \frac{\beta_{\rm D}}{\mu p}\frac{\sigma^2}{\left(\pi\lambda_{\rm D}\right)^{\alpha/2}},
\end{equation}
\normalsize
and $E_{\rm C}=\mathbb{E} \left \{ P_{\rm C}^{2/{\alpha}} \right \}$, $E_{\rm D}=\mathbb{E} \left \{ P_{\rm D}^{2/{\alpha}} \right \}$. Moreover, we have
\small
\begin{equation} \label{ED}
E_{\rm D}\!\!=\!\!\int_0^{\infty}\!\!\left(\!A_1\!\left(\!A_2\!+\!E_{\rm D}\!\right)\!x^{-1}\!\! +\!\! A_3 x^{-2/\alpha}\!\right)\!{\rm e}^{-\frac{\alpha}{2}A_1\!\left(\!A_2+E_{\rm D}\right)x^{2/\alpha}-A_3 x}\mathrm{d}x,
\end{equation}
\normalsize
where
\begin{equation*}
A_1\! =\! \left(\!\frac{\beta_{\rm D}}{\mu}\!\right)^{\!\!2/{\alpha}}\!\!\!\!\!\!\frac{2/\alpha}{\mathrm{sinc}\left(2/{\alpha}\right)}, ~~
A_2 \!=\! \left( \!\frac{\lambda_{\rm C}}{\lambda_{\rm D}} \!\right )\!E_{\rm C},~~
A_3 \!=\! \frac{\beta_{\rm D}\sigma^2}{\mu\left(\pi \lambda_{\rm D}\right)^{2/\alpha}}.
\end{equation*}

\end{theorem}
 
\begin{proof}
The proof is given in Appendix \ref{App1}.
\end{proof}

We call the derived CDF in Theorem \ref{Theorem1} an equilibrium distribution, since it is an equation to which all nodes converge through a dynamic adaptation algorithm expressed in \eqref{UpdateFM}. The convergence issue is carried out in \cite{Trans_Foschini_ASimple}, which necessitates the suitable selection of $\gamma$. In real communication networks, the power allocated to each UE is limited. Extending Theorem \ref{Theorem1} to the constrained power scenario is straightforward. Assume $P_{\rm max}$ is the maximum value for $P_{\rm D}$ and define $P_{\rm D}^{\rm c}=\min\left\{P_{\rm D},P_{\rm max}\right\}$, where the superindex `c' denotes the constrained power case. The new CDF of D2D transmit power can be written as
\begin{equation} \label{EqDistD2DConst}
\mathbb{P}\left\{P_{\rm D}^{\rm c}\leq p\right\} =
\begin{cases}
\mathbb{P}\left\{P_{{\rm D}}\leq p\right\}, & p < P_{\rm max},\\ 
1, & p \geq P_{\rm max}.
\end{cases}
\end{equation}

A closed form solution for the D2D power distribution in the interference-limited case $\left(\sigma^2=0\right)$ is given in Corollary \ref{Corollary1}.
\begin{corollary} \label{Corollary1}
In an interference-limited scenario, the CDF of the D2D transmit power is simplified to
\begin{equation} \label{Eq.Cor1}
\mathbb{P}\left\{P_{{\rm D}}\leq p\right\} = \frac{1}{1+(\frac{\beta_{\rm D}}{\mu p})^{2/{\alpha}} \frac{1}{\mathrm{sinc}(2/{\alpha})}(E_{\rm D}+(\frac{\lambda_{\rm C}}{\lambda_{\rm D}})E_{\rm C})}.
\end{equation}
\end{corollary}
It should be noted that the CDF goes to $0$ for $p\to 0$ or $\beta_{\rm D} \to \infty$, and to $1$ for $p\to \infty$ or $\beta_{\rm D} \to 0$, as expected.

In relatively lossy environments $\left(\alpha = 4\right)$, the integral of Theorem \ref{Theorem1} can be expressed in a more compact form using the standard complementary error function \cite{Trans_Chiani_New}:
\begin{equation}
\mathrm{erfc}\left(x\right)=\frac{2}{\sqrt{\pi}}\int_x^{\infty}{\rm e}^{-u^2}\,\mathrm{d}u \approx \frac{1}{6}{\rm e}^{-x^2}+\frac{1}{2}{\rm e}^{-4x^2/3}.
\end{equation}
The result is stated as follows.

\begin{corollary} \label{Corollary2}
In a relatively lossy environment with $\alpha=4$, the CDF of the D2D transmit power is simplified to
\small
\begin{equation} \label{Eq.Cor.4.2}
\mathbb{P}\left\{P_{{\rm D}}\leq p\right\} \!=\! \frac{1}{2}\sqrt{\frac{\pi}{k'_2}}{\rm e}^{\frac{{k'_1}^2}{4k'_2}}\mathrm{erfc}\!\!\left(\!\!\sqrt{\frac{{k'_1}^2}{4k'_2}}\right) \!\!\approx\!\frac{1}{12}\sqrt{\!\frac{\pi}{k'_2}}\left(\!1\!+\!3{\rm e}^{-\frac{{k'_1}^2}{12k'_2}} \right)\!,
\end{equation}
\normalsize
where the constants $k'_1$ and $k'_2$ are defined as
\begin{equation*}
k'_1 = 1+\frac{\pi}{2}\sqrt{\frac{\beta_{\rm D}}{\mu p}}\left(E_{\rm D}+\left(\frac{\lambda_{\rm C}}{\lambda_{\rm D}}\right) E_{\rm C}\right), ~~
k'_2 = \frac{\beta_{\rm D}}{\mu p} \frac{ \sigma^2}{\left(\pi\lambda_{\rm D}\right)^2}.
\end{equation*}
\end{corollary}

Proposition \ref{Proposition1} gives a tight upper-bound for the CDF of the D2D transmit power.

\begin{proposition} \label{Proposition1}

The CDF of the D2D transmit power expressed in \eqref{EqDistD2D} is upper-bounded as
\begin{equation} \label{EqUpperD2D}
\mathbb{P}\left\{P_{{\rm D}}\leq p\right\} \! \leq \! \left(\frac{1}{u k_1}\right)^{\!\!1/u} \!\!\left(\Gamma\! \left(\frac{2}{\alpha}+1\right)\!\left(\frac{1}{v k_2}\right)^{\!\frac{2}{\alpha}}\right)^{\!\!1/v}\!\!,
\end{equation}
where $\Gamma(t)=\int_0^{\infty}x^{t-1}{\rm e}^{-x}\,\mathrm{d}x$ is the Gamma function and $u$ and $v$ are constants satisfying the conditions below.
\begin{itemize}
\item $\frac{1}{u}+\frac{1}{v}=1,$
\item $u^{\alpha}-2u^{\alpha-1}+u^{\alpha-2}=\frac{\left(\frac{k_2}{{\rm e}}\right)^2}{\left(\frac{k_1}{{\rm e}}\Gamma\left(\frac{2}{\alpha}+1\right)\right)^{\alpha}}.$
\end{itemize}
\end{proposition}

\begin{proof}
The proof is given in Appendix \ref{App2}.
\end{proof}

Via simulation, we show that the derived upper-bound is so tight that it can be used as a good approximation for \eqref{EqDistD2D}.

\section{Simulation Results} \label{4.Sim}
For numerical simulations, we consider a hexagonal cell in which D2D users are uniformly and independently distributed \cite{Book_Stoyan_Stochastic}. Initially, each D2D transmitter is associated to the nearest D2D receiver in its proximity. For the considered scenario, the inter-site distance is $500$\,m, the thermal noise is $-100$\,dB, the path loss is defined by $30.6+40 \log(d)$, where $d$ is measured in meters, and $P_{\rm max} =23 ~{\rm dBm}$. 
The power control method is the Foschini-Miljanic algorithm \cite{Trans_Foschini_ASimple} which converges after about $50$ iterations for $\gamma=0.06$ in the considered setup.


\begin{figure}[!t]
   \centering
   \includegraphics[width=0.5\textwidth]{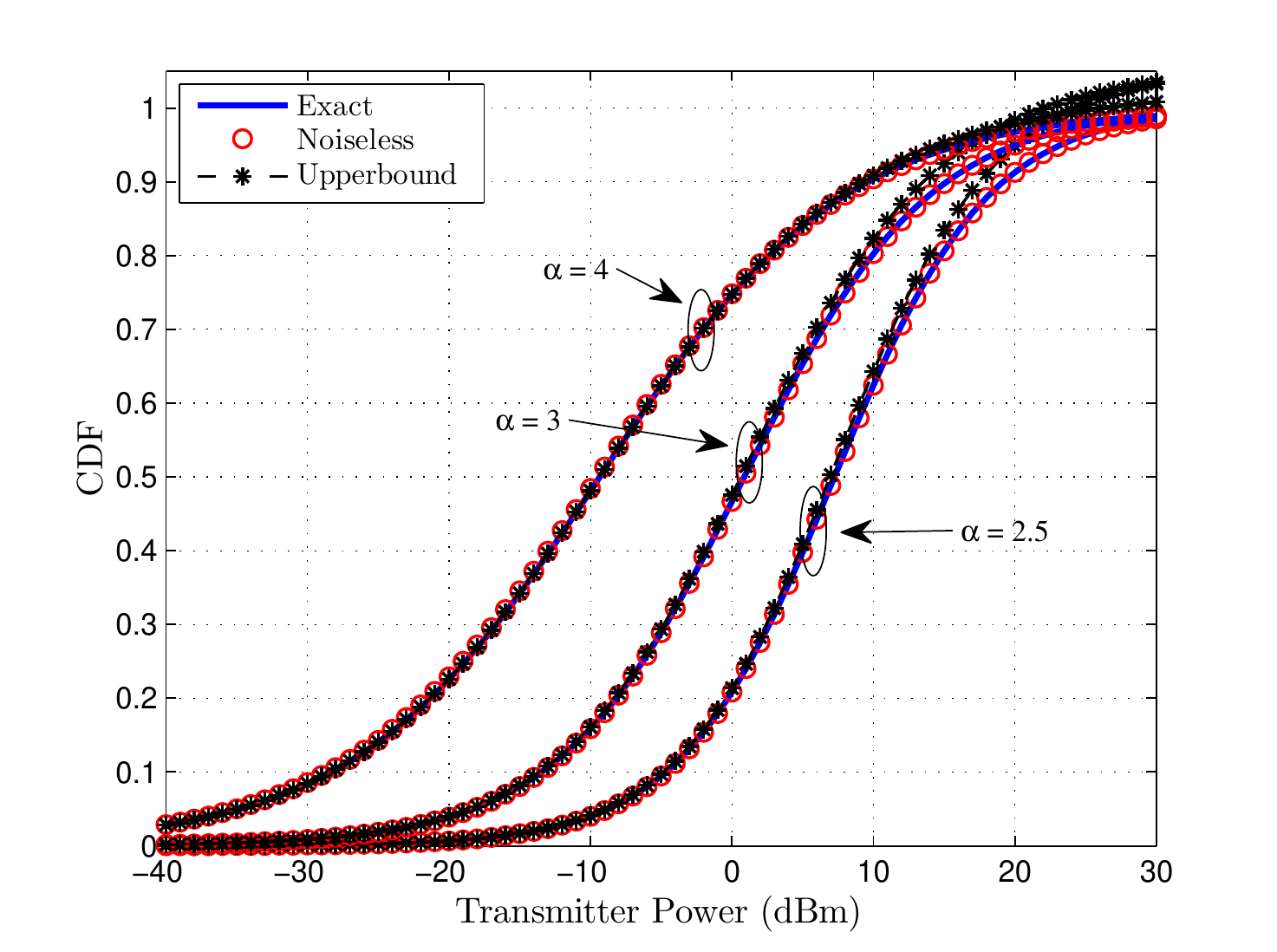}
   \vspace{-0.7cm}
   \caption{Comparison between the exact power CDF, its interference-limited case, and its derived upper-bound with  $\beta_{\rm D}=-10~{\rm dB}$ and various $\alpha$.}
   \vspace{-0.5cm}
   \label{D2DExacUpperNN}
\end{figure}

Fig. \ref{D2DExacUpperNN} shows the analytical converged CDF of D2D transmit power levels along with the special interference-limited case in \eqref{Eq.Cor1} and the derived upper-bound in \eqref{EqUpperD2D}. As seen in this figure, the assumption of ignoring noise in such networks is well justified and the upper-bound is demonstrated to be a good approximation.

In Fig. \ref{VarAlpha}, we show the analytical and simulated power CDF for different values of $\alpha$. In denser regions, i.e., for larger values of $\alpha$, the number of users with an allocated power of $p_0$ or less increases, as well. In other words, in order to manage interference in more populated regions, lower transmitting powers should be allocated to users.

\begin{figure}[!t]
   \centering
   \includegraphics[width=0.5\textwidth]{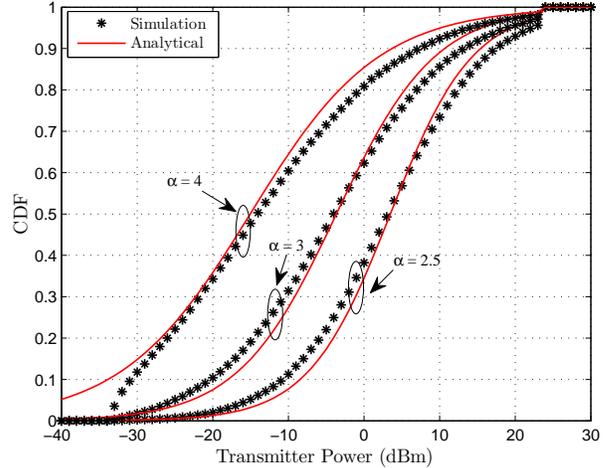}
\vspace{-0.7cm}
   \caption{Comparison between analytical and simulation results of D2D power CDF with $\beta_{\rm D}=-10~{\rm dB}$, $\mu=10^{-3}$, and various $\alpha$.}
\vspace{-0.5cm}
   \label{VarAlpha}
\end{figure}

The analytical and simulated power CDFs of D2D transmitters for a fixed $\alpha$ and various values of SINR threshold $\beta_{\rm D}$ is shown in Fig. \ref{VarThreshold}. For a specific transmission power level $p_0$, as the SINR threshold increases, fewer users can satisfy the SINR constraint, and thus, the probability that the power level at each D2D link is smaller than $p_0$ will decrease. Moreover, increasing the SINR threshold causes more users to reach the maximum power.

\begin{figure}[!t]
   \centering
   \includegraphics[width=0.5\textwidth]{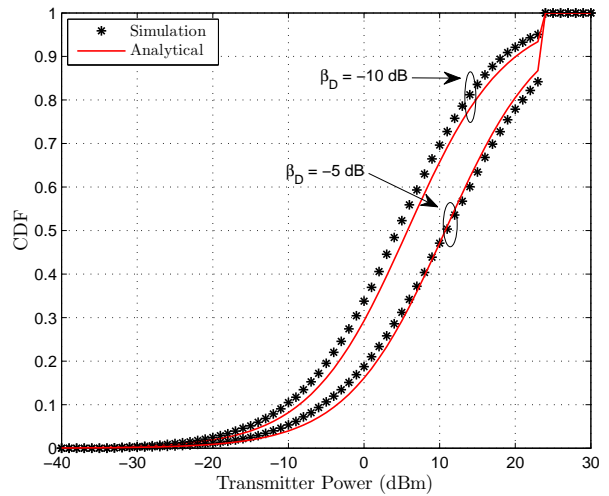}
\vspace{-0.7cm}
   \caption{Comparison between analytical and simulation results of D2D power CDF with $\alpha = 3$, $\mu=10^{-4}$, and different D2D SINR thresholds $\beta_{\rm D}$.}
\vspace{-0.5cm}
   \label{VarThreshold}
\end{figure}

\section{Conclusion} \label{Conclusion}
The analysis of power distribution in a D2D-underlaid cellular network in the uplink transmission is the focus of this letter. We considered an underlay scenario for spectrum sharing. Hence, D2D users may cause severe interference to the macro-cellular UEs and vice versa. Using stochastic geometry, we derived the distribution of D2D transmit power levels. Two special cases of the power CDF, i.e., the interference-limited case $\left(\sigma^2=0\right)$ and the relatively lossy environments $\left(\alpha=4\right)$ are analyzed. In addition, we derived a tight closed-form upper-bound for the D2D power CDF. Simulations showed a good correspondence between the analytical and numerical results, which corroborates the validity of our analysis.

\appendices

\section{Proof of Theorem \ref{Theorem1}} \label{App1}
From \eqref{EqSINR}, the CDF of $P_{\rm D}$ can be written as
\begin{equation*}
\mathbb{P}\left\{P_{{\rm D}}\leq p\right\} = \underset{x}{\mathbb{E}}\left \{{\rm e}^{-\sigma^2 \frac{\beta_{\rm D}\left\|x_{0,{\rm D}}\right\|^{\alpha}}{\mu p}} \underset{I}{\mathbb{E}}\left\{{\rm e}^{-I\frac{\beta_{\rm D}\left\|x_{0,{\rm D}}\right\|^{\alpha}}{\mu p}}\right\}\right\},
\end{equation*}
where $I=I_{\rm C}+I_{\rm D}$. Since the two nodes of a D2D pair are the nearest devices to each other, the probability density function (PDF) of $\left\|x_{0,{\rm D}}\right\|$ is given by \cite{Trans_Andrews_Tractable}
\begin{equation} \label{NearestDist}
f_{\left\|x_{0,{\rm D}}\right\|}(x)=2\pi \lambda_{\rm D} x {\rm e}^{-\pi \lambda_{\rm D} x^2}.
\end{equation}
Hence, we can write
\begin{equation*}
\mathbb{P}\left\{P_{{\rm D}}\leq p\right\} = \int_0^{\infty} \!\! 2\pi \lambda_{\rm D} x\, {\rm e}^{-\pi \lambda_{\rm D} x^2} {\rm e}^{-\frac{\beta_{\rm D} \sigma^2}{\mu p}x^{\alpha}} \mathcal{L}_I \left(\frac{\beta_{\rm D} x^{\alpha}}{\mu p}\right)\! \mathrm{d}x,
\end{equation*}
where $\mathcal{L}$ is the Laplace transform operator. Since $I_{\rm C}$ and $I_{\rm D}$ are independent, we have $\mathcal{L}_I \left(s\right) = \mathcal{L}_{I_{\rm C}} \left(s\right)\mathcal{L}_{I_{\rm D}} \left(s\right)$ and
\begin{equation} \label{IDCompute1}
\mathcal{L}_{I_{\rm C}} \left(s\right) = \mathbb{E}\left \{\prod\nolimits_{i\in \Phi_{\rm C}} {\rm e}^{-s P_{i,{\rm C}}\left|h_{i,{\rm C}}\right|^2\left\|x_{i,{\rm C}}\right\|^{-\alpha}} \right \}.
\end{equation}
Note that the transmitting and fading powers are independent of the location of the nodes, and thus, expectations over them can move into the product. This yields
\begin{align} \label{LaplaceIC}
\mathcal{L}_{I_{\rm C}}(s) &= \underset{x}{\mathbb{E}} \left \{ \prod\nolimits_{i\in \Phi_{\rm C}} \underset{P}{\mathbb{E}} \left \{ \frac{1}{1+sP_{i,{\rm C}}\left\|x_{i,{\rm C}}\right\|^{-\alpha}} \right \}  \right \}  \nonumber \\
& \overset{(a)}{=} \exp \left \{ -\lambda_{\rm C} \int_{\mathbb{R}^2} \left [1-\underset{P}{\mathbb{E}} \left \{ \frac{1}{1+sP_{i,{\rm C}} x^{-\alpha}} \right \} \right ]\, \mathrm{d}x \right \} \nonumber \\
& \overset{(b)}{=} \exp \left \{-\frac{1}{\mathrm{sinc}\left(2/{\alpha}\right)} \pi \lambda_{\rm C} s^{2/{\alpha}} \mathbb{E} \left \{ P_{\rm C}^{2/{\alpha}} \right \} \right \},
\end{align}
where $(a)$ follows from the Campbell's theorem for PPPs \cite{Book_Kingman_Poisson} and $(b)$ follows from the fact that expectation over $P$ and integration over $x$ are interchangeable. For the D2D interference, the Laplace transform of $I_{\rm D}$ is not exactly similar to \eqref{LaplaceIC}. However, we can approximate it as
\begin{align} \label{LaplaceID}
\mathcal{L}_{I_{\rm D}}(s) &= \exp \left \{ -\lambda_{\rm D} \int_{\mathcal{A}} \left [1-\underset{P}{\mathbb{E}} \left \{ \frac{1}{1+sP_{i,{\rm D}} x^{-\alpha}} \right \} \right ]\, \mathrm{d}x \right \} \nonumber \\
& \overset{(c)}{\approx} \exp \left \{-\frac{1}{\mathrm{sinc}\left(2/{\alpha}\right)} \pi \lambda_{\rm D} s^{2/{\alpha}} \mathbb{E} \left \{ P_{\rm D}^{2/{\alpha}} \right \} \right \},
\end{align}
where $\mathcal{A}={\mathbb{R}^2\backslash \mathcal{B}(\mathcal{O},\left\|x_{0,{\rm D}}\right\|)}$ and $\mathcal{B}\left(\mathcal{O},r\right)$  denotes a circle with center at $\mathcal{O}$ and radius of $r$. The approximation in $(c)$ is for ignoring the small effect of integration from $0$ to $\left\|x_{0,{\rm D}}\right\|$.

We then have
\small
\begin{equation*}
\mathcal{L}_I \left(s\right) = \exp \left \{-\frac{\pi s^{2/{\alpha}}}{\mathrm{sinc}\left(2/{\alpha}\right)} \left ( \lambda_{\rm C} \mathbb{E} \left \{ P_{\rm C}^{2/{\alpha}} \right \} + \lambda_{\rm D} \mathbb{E} \left \{ P_{\rm D}^{2/{\alpha}} \right \}\right ) \right\}.
\end{equation*}
\normalsize
Hence, the D2D power distribution turns out to be
\begin{align*}
\mathbb{P}\left\{P_{{\rm D}}\leq p\right\} =\int_0^{\infty} &2\pi \lambda_{\rm D} x {\rm e}^{-\pi \lambda_{\rm D} x^2} {\rm e}^{-\frac{\beta_{\rm D}}{\mu p} \sigma^2 x^{\alpha}} \nonumber \times  \\
&  {\rm e}^{-\frac{\pi}{\mathrm{sinc}\left(2/{\alpha}\right)} \left(\frac{\beta_{\rm D}}{\mu p}\right)^{2/\alpha}\left(\lambda_{\rm C} E_{\rm C}+\lambda_{\rm D} E_{\rm D}\right) x^2} \, \mathrm{d}x.
\end{align*}
The result expressed in \eqref{EqDistD2D} is obtained using $\pi \lambda_{\rm D} x^2 \rightarrow r$.

There is a \emph{circularity} in \eqref{EqDistD2D}, in which $E_{\rm D}$ depends on the PDF of $P_{\rm D}$, which is obtained from its CDF. In fact, we have
\begin{equation*}
E_{\rm D}=\mathbb{E}\left\{ P_{\rm D}^{2/{\alpha}}\right\}=\int_0^{\infty}x^{2/\alpha}f_{P_{0,{\rm D}}}(x)\,\mathrm{d}x.
\end{equation*}
Differentiating the CDF \eqref{EqDistD2D}, changing the order of integrations, and some changes of variables, we get \eqref{ED}. Since the analytical solution of \eqref{ED} is not explicit, it should be solved numerically.

\section{Proof of Proposition \ref{Proposition1}} \label{App2}
From Hölder's inequality \cite{Book_Jeffrey_Table}, eqn. (12.312), if $\left|f\left(x\right)\right|^u$ and $\left|g\left(x\right)\right|^v$ are two real-valued integrable functions on $\left[0,\infty\right]$ with $u>1$ and $\frac{1}{u}+\frac{1}{v}=1$, then
\small
\begin{equation*}
\int_0^{\infty}\!\!\!f\left(x\right)\!g\left(x\right)\,\mathrm{d}x\leq \left ( \int_0^{\infty} \left|f\left(x\right)\right|^u\,\mathrm{d}x \right )^{\!\!1/u}\!\! \left ( \int_0^{\infty} \left|g\left(x\right)\right|^v\,\mathrm{d}x \right )^{\!\!1/v}.
\end{equation*}
\normalsize
Letting $f\left(x\right)={\rm e}^{-a_i x}$ and $g\left(x\right)={\rm e}^{-b_i x^{\alpha/2}}$, we have  \cite{Book_Jeffrey_Table}, eqn. (3.326.1)
\small
\begin{equation} \label{Holder}
\int_0^{\infty}\!\!\!{\rm e}^{-a_i x-b_i x^{\alpha/2}}\mathrm{d}x \leq \left(\frac{1}{ua_i}\right)^{\!\!1/u}\!\! \left(\Gamma \!\left(\frac{2}{\alpha}+1\right) \left(\frac{1}{vb_i}\right)^{\frac{2}{\alpha}}\right)^{\!\!1/v}\!\!\!\!\!.
\end{equation}
\normalsize
Differentiating the right-hand side of \eqref{Holder} with respect to $u$ and setting the result to $0$, the second condition follows.

\bibliographystyle{IEEEtran}
\bibliography{PowDistD2D}
\end{document}